\documentclass[runningheads]{llncs}

\usepackage[title]{appendix}
\usepackage[T1]{fontenc}

\usepackage{graphicx} 
\usepackage{amsmath, amssymb}
\usepackage{mathtools}
\usepackage{complexity}
\usepackage{algorithm}
\usepackage[noend]{algpseudocode}
\usepackage{float}
\usepackage{tikz}
\usepackage{comment}
\usepackage[normalem]{ulem}
\usetikzlibrary{arrows,automata,backgrounds,arrows.meta, positioning,fit,calc}

\usepackage{amsthm}

 \newtheorem*{theorem*}{Theorem}
  \newtheorem*{lemma*}{Lemma}
    \newtheorem*{corollary*}{Corollary}

\newcommand{\Set}[2]{\left\{\,{#1} : {#2}\,\right\}}

\newcommand{\intr}{\mathrel{\rhd\mspace{-10mu}\lhd}}
\newcommand{\intro}{\intr_{\!_\preceq} \!}

\newcommand{\La}{\mathcal{L}}

\newcommand{\pf}[1]{\operatorname{Pref}({#1})}

\newcommand{\mt}[1]{\mathcal{#1}} 

\usepackage{xcolor}

\newlang{\Wh}{W}
\newlang{\UW}{UW}
\newlang{\EW}{EW}
\newlang{\DEF}{DEF}
\newlang{\RDEF}{RDEF}
\newlang{\SLT}{SLT}
\newlang{\Fin}{Fin}
\newlang{\CoFin}{CoFin}

\newlang{\SETH}{SETH}
\newlang{\OV}{OV}
\newlang{\YES}{YES}

\begin{document}
\title{Universally Wheeler Languages}

\author{Ruben Becker\inst{1}\orcidID{0000-0002-3495-3753} \and
Giuseppa Castiglione\inst{2}\orcidID{0000-0002-1838-9785} \and
Giovanna D'Agostino\inst{3}\orcidID{0000-0002-8920-483X} \and
Alberto Policriti\inst{3}\orcidID{0000-0001-8502-5896} \and
Nicola Prezza\inst{1}\orcidID{0000-0003-3553-4953} \and
Antonio Restivo\inst{2}\orcidID{0000-0002-1972-6931} \and
Brian Riccardi\inst{4}\orcidID{0000-0002-4925-9529}
}%

\authorrunning{R. Becker et al.}

\institute{Dept. of Environmental Sciences, Informatics and Statistics, Ca' Foscari University of Venice, Italy \and
Dept.\ of Mathematics and Computer Science, University of Palermo, Italy \and
Dept.\ of Mathematics, Computer Science and Physics, University of Udine, Italy \and
Dept.\ of Informatics, Systems and Communication, University of Milano-Bicocca, Italy}
\maketitle              % typeset the header of the contribution
\begin{abstract}
The notion of Wheeler languages is rooted in the Burrows-Wheeler transform (BWT), one of the most central concepts in data compression and indexing. The BWT has been generalized to finite automata, the so-called Wheeler automata, by Gagie et al.~[Theor.\ Comput.\ Sci.\ 2017]. Wheeler languages have subsequently been defined as the class of regular languages for which there exists a Wheeler automaton accepting them. Besides their advantages in data indexing, these Wheeler languages also satisfy many interesting properties from a language theoretic point of view [Alanko et al., Inf.\ Comput.\ 2021]. A characteristic yet unsatisfying feature of Wheeler languages however is that their definition depends on a \emph{fixed} order of the alphabet. 

In this paper we introduce the \emph{Universally Wheeler languages} $\UW$, i.e., the regular languages that are Wheeler with respect to \emph{all} orders of a given alphabet. Our first main contribution is to  relate $\UW$ to some very well known regular language classes.
We first show that the \emph{Striclty Locally Testable} languages are strictly included in $\UW$. After noticing that $\UW$ is not closed under taking the complement, we prove that the class of languages for which both the language and its complement are in $\UW$ exactly coincides with those languages that are \emph{Definite} or \emph{Reverse Definite}\\
Secondly, we prove that deciding if a regular language  given by a DFA  is in $\UW$ can be done in quadratic time. We also show that this is optimal unless the Strong Exponential Time Hypothesis (\SETH{}) fails.

\end{abstract}

\section{Introduction}
The probably most powerful single algorithmic primitive for indexing data is sorting. The Burrows-Wheeler transform (BWT) is a string transformation that is based on sorting all rotations of a given string. It is at the base of the most efficient text indexes developed in the last three decades such as, for example, the FM index~\cite{FerraginaM05} and the $r$-index~\cite{GagieNP20}. 
This technique of indexing compressed data based on sorting has been extended to sets of strings that can be described by finite automata, i.e., regular languages, by Gagie et al.~\cite{GAGIE201767}. The class of automata that permits sorting its states so that the underlying language can be indexed was termed \emph{Wheeler automata}. Subsequently, there has been an interesting line of works~\cite{AlankoDPP21,becker2023optimal,jacm/CotumaccioDPP23} that aims at understanding the class of regular languages accepted by such automata, the so-called \emph{Wheeler languages}.

The definition of Wheeler languages is, however, characterized by the peculiarity that it depends on a fixed total order of the alphabet. This peculiarity is inherited from the definition of Wheeler automata as follows. Given a totally ordered alphabet, a language is called Wheeler if there exists an automaton accepting that language, which is Wheeler with respect to this given order. This dependency of an alphabet order is unsatisfying for several reasons. Firstly, from a practical point of view, in terms of the original indexing problem, this alphabet order is not very meaningful as any order of the alphabet is equally interesting, e.g. after a suitable alphabet remapping. Secondly, from a more theoretical perspective, if we want to understand the characteristics of Wheeler languages by relating them to other classes of regular languages, those classes typically do not depend on alphabet orders and thus such a relation is not easily possible. 

There are two canonical ways of removing this dependency of an alphabet order: to consider the languages for which there \emph{exists} an alphabet order that makes 
the language Wheeler or to consider those languages that are Wheeler with respect to \emph{all} alphabet orders. The first approach was followed by D'Agostino et al.~\cite{ictcs/Martincigh21} leading to the class of \emph{Existentially Wheeler} languages $\EW$ (named \emph{Generalized Wheeler} in \cite{ictcs/Martincigh21}). However, their approach is unsatisfying from a computational viewpoint. While the Wheeler language recognition problem, i.e., the problem of checking if a language $\La$ is Wheeler given a DFA accepting it, can be solved in quadratic time in the size of the DFA~\cite{becker2023optimal}, the problem of deciding whether $\La$ is in $\EW$ is NP-complete~\cite{ictcs/Martincigh21},  even when the language is given by a DFA.

\paragraph{Our contribution.}
In this article we follow the second approach by defining the class of \emph{Universally Wheeler languages} $\UW$, i.e., those languages that are Wheeler with respect to every order of the alphabet. We present two main contributions in this area. (1)~We manage to relate $\UW$ to some very well known classes of regular languages, namely the \emph{Definite} ($\DEF$), \emph{Reverse Definite} ($\RDEF$), and the \emph{Strictly Locally Testable} ($\SLT$) languages (Section~\ref{sec:DefRevDef}).
More precisely, we first show that $\SLT$ is strictly included in $\UW$. As $\UW$ is not closed under taking the complement, we then consider the class of languages $\La$ such that both $\La$ and $\overline{\La}$ belong to $\UW$ and prove that this class coincides with $\DEF \cup \RDEF$.  
(2)~From a computational point of view, we prove that $\UW$ behaves better than $\EW$:  checking whether the language accepted by a given DFA is in \UW{} can be done in quadratic time in the size of the DFA (Section~\ref{sec:decidingUW}). We complete this algorithmic result with a fine-grained complexity lower bound showing that no truly subquadratic time algorithm is possible unless the Strong Exponential Time Hypothesis (\SETH{}) fails.  \\
Due to space limitations, some  of the proofs  are moved to  the appendix.

\section{Preliminaries}\label{sec:prelim}

\paragraph{Automata.}
A \emph {deterministic finite automaton} (DFA) is  a  tuple $\mt D= (Q, s,\Sigma, \delta, F)$, where  $Q$ is a finite  non-empty set of \emph{states}, $s$ is the 
\emph{initial state} (or \emph{source}), $\delta\colon Q \times \Sigma \rightarrow Q$ is the (possibly partial)  
\emph{transition function}, and $F \subseteq Q$ is the set of \emph{final states}. 

If $\delta$ is not defined on $(q,\sigma)$ we write $\delta(q,\sigma)=\bot$.  If $\delta$ is a total function we say that the DFA is \emph{complete}.  We sometimes    omit the alphabet $\Sigma$ and simply write a  DFA as a tuple $\mt D= (Q, s, \delta, F)$. If we are dealing with a single DFA,    the  letters $Q,s,\delta,F$   always refer to the set of $\mt  D$-states, the $\mt  D$-initial state etc.   Sometimes we shall describe the transition function $\delta$ using triples (edges), where $(q,q',a)$ stands for $q'=\delta(q,a)$.

We denote with $|\delta|$ and $|Q|$ the number of edges and states of an automaton $\mt D= (Q, s, \delta, F)$, and with $|\mt D| = |Q| + |\delta|$ its total size. 

As customary, we extend $\delta$ to strings as follows: for all $q\in Q$, $a\in \Sigma$ and $\alpha \in \Sigma^*$, we define
$\delta(q,\varepsilon) = q$ and $\delta(q,\alpha a)= \delta(\delta(q,\alpha),a)$, 
where $\delta(\bot,a)=\bot$.

Given a DFA $\mt D= (Q, s,\Sigma, \delta, F)$ and any of its states $q\in Q$, we denote with $I_q$ the set of all strings labeling paths starting from the source of $\mt  D$ and ending in $q$, i.e., $I_q = \{\alpha \in \Sigma^*\ :\ \delta(s, \alpha) = q\}$. The language \emph{accepted} (or \emph{recognized}) by $\mt  D$ is then defined as $\La(\mt D)=\cup_{q\in F} I_q$. 
Moreover, for any \(q \in Q\), we denote by $\lambda(q)$ the set of labels of all transitions whose endpoint is $q$, i.e.,

$\lambda(q) \coloneq  \Set{a \in \Sigma}{(\exists p \in Q)(\delta(p, a) = q)}$.

We mostly  consider    \emph{trimmed} DFAs, where     every state is reachable from the initial state and can reach at least one final state. This is not restrictive:  every automaton can be turned into a trimmed and equivalent one by simply deleting unreachable states as well as states not reaching at least one final state. 
For trimmed automata the following hold:
 there can be at most one state without incoming edges, namely $s$, and  
every string that can be read starting from $s$ belongs to the set of prefixes, $\pf{\La}$, of the language $\La$. 

The languages accepted by automata form the class of \emph{regular languages}. 
Given a regular language $\La$, there exists a unique (up to isomorphism) state-wise minimum \emph{complete} DFA recognizing $\La$.
The states of this minimum DFA correspond to the classes of the Myhill-Nerode equivalence relation $\equiv_\La$ on $\Sigma^*$,  defined  as follows:
\[
    \alpha \equiv_\La \beta \iff \{\gamma\in \Sigma^* \mid \alpha\gamma \in \La\}= \{\gamma\in \Sigma^* \mid \beta\gamma \in \La\}.
\]

We denote by $\mt D^c_\La$ the \emph{complete} DFA   having as set of states  the  classes $[\alpha]_\La$  of the   Myhill-Nerode equivalence,  $s=[\varepsilon]_\La$,  
      $  
        \delta([\alpha]_\La,a)=[\alpha\cdot a]_\La,
       $  
and  $F=\{[\alpha]_\La  \mid\,\alpha\in \La\}$.
Then it can be proved that  $\mt D^c_\La$ is the    DFA   with minimum number of states among all complete DFA recognizing $\La$ and it is unique up to isomorphism. 
Let  $\mt D_\La$ be  the DFA obtained from   $\mt D^c_\La$  by considering only    the classes  $[\alpha]_\La$ with   $\alpha \in \pf{\La}$ and the transitions among them.  Then one can easily check that $\mt D_\La$ recognizes $\La$, it is trimmed,  and can differ from $\mt D^c_\La$ for (at most) one non-final \emph{absorbing} state --- that is, a state \(q\in Q\) such that \(\delta(q,a)=q\), for all \(a \in \Sigma\).
The automaton \(\mt D_\La\)   is the DFA with minimum number of states among  (complete or not) DFAs recognizing $\La$. 
  \paragraph{Orders.}
If $(Z,\leq) $ is a partial order, we denote by $(Z,<) $ its corresponding strict partial order. 

If $ \preceq $ is a total order on the alphabet  $ \Sigma $, we extend  it  to strings in $\Sigma^*$ \emph{co-lexicographically}, that is, for $ \alpha, \beta \in \Sigma^* $, we have \index{alpha is smaller than beta@ $ \alpha \preceq \beta $} $ \alpha \preceq \beta $ if and only if either $ \alpha $ is a suffix of $ \beta $ (denoted by $\alpha \dashv \beta$), or there exist $ \alpha', \beta', \gamma \in \Sigma^* $ and $ a, b \in \Sigma $, such that $ \alpha = \alpha' a \gamma $ and $ \beta = \beta' b \gamma $ and $ a \prec b $.

\paragraph{Wheeler Automata and Wheeler Languages.}

Wheeler Automata are a special class of DFAs that leverage an a priori fixed order of the alphabet in order to achieve, among other things, efficient compression and indexing (see \cite{GAGIE201767}).

\begin{definition}[Wheeler Automaton]
\label{def:wheeler automaton}
\sloppy Let $\preceq$ be a fixed  total order of the alphabet $\Sigma$. A  Wheeler DFA (WDFA for brevity) is a trimmed DFA endowed with a total order $(Q, \leq)$ on
the set of states such that the initial state $s$ has no incoming transitions, it is
$\leq$-minimum, and such that the following two \emph{Wheeler axioms} are satisfied. Let $p' = \delta(p, a)$ and $q' = \delta(q, b)$:
\begin{enumerate} 
    \item[] (W1) if $a \prec b $, then $ p' < q'$; 
    \item[] (W2) if $a  = b$, $p < q$, and  $p' \neq q'$, then $ p' < q'$.
\end{enumerate}
\end{definition}

Notice that (W1) implies that a WDFA is \emph{input consistent}, that is, 
$|\lambda(p)|=1$ for all states $p\neq s$.  Any DFA can be transformed into an equivalent input consistent DFA in $\mt O\big(|Q|\cdot |\Sigma|\big)$ time by simply creating, for each state $q\in Q$, at most $|\Sigma|$ copies of $q$, one for each different incoming label of $q$.

\begin{figure}[ht]%
\begin{center}

\begin{tikzpicture}[->,>=stealth', semithick, initial text={}, auto, scale=.25]
 \node[state, label=above:{}, initial] (0) at (0,0) {$s$};
 \node[state, label=above:{}, accepting] (1) at (5,3) {$q_1$};
 \node[state, label=above:{}, accepting] (2) at (10,3) {$q_2$};
 \node[state, label=above:{}] (4) at (5,-3) {$q_4$};
 \node[state, label=above:{}] (3) at (10,-3) {$q_3$};
 \node[state, label=above:{}, accepting] (5) at (15,-3) {$q_5$};

\draw (0) edge [above] node [above] {$a$} (1);
\draw (1) edge [below] node [above] {$c$} (2);
\draw (2) edge [loop above] node [above] {$c$} (N1);
\draw (0) edge [bend left=0, above] node [bend right, above] {$d$} (4);
\draw (4) edge [above] node [above] {$c$} (3);
\draw (3) edge [loop below] node [above, xshift=8] {$c$} (3);
\draw (3) edge[above] node [above] {$f$} (5);

\end{tikzpicture}
 
\end{center}
    \caption{  A WDFA $\mt D$ recognizing the language $\mt L_d = ac^*\cup dc^+f$. The only order that   makes $\mt D$ Wheeler is $s < q_1 < q_2 < q_3 < q_4 < q_5$.}
    \label{fig:example WDFA}
\end{figure}
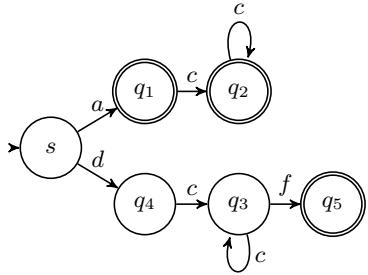

It can be proved (see \cite{AlankoDPP21,jacm/CotumaccioDPP23}) that, if such a total order satisfying Definition~\ref{def:wheeler automaton} exists, then it is unique.  More precisely, given an order $\preceq$ on $\Sigma$ and  a trimmed, input consistent   DFA $\mt D$ in which the initial state has no incoming edges,     we can   define  a partial order $\leq_\mt D$  on $Q$   by
$q\leq _\mt D q' \Leftrightarrow (q=q') \lor   \forall \alpha\in I_q, \forall \beta\in I_{q'} ~(\alpha \prec \beta)$. 
Then it can be proved that  $ \leq_\mt D$ always satisfies properties (W1) and (W2) and that 
  the     DFA is Wheeler precisely when the partial order $ \leq_\mt D$ is total.  
In this case,  we can decide whether $q \leq_\mt D q'$ by simply checking the relative co-lexicographical  order among a pair $(\alpha, \beta)$ with $\alpha\in I_q$ and $\beta\in I_{q'}$. 

\medskip

Switching from automata to languages we define Wheeler languages~\cite{AlankoDPP21}.
\begin{definition}\label{def:WL}
    A language is \emph{Wheeler} if there exists a WDFA recognizing it. 
\end{definition}
One  can  easily generalize Definition~\ref{def:wheeler automaton} to  WNFA     and 
  prove that WDFA and WNFA have the same expressive power  in terms of accepted languages (see \cite[Theorem 3.16]{AlankoDPP21}).
We denote by $\Wh(\preceq)$ the class of all regular languages that are Wheeler w.r.t.\ the order $\preceq$.
In \cite{becker2023optimal} a characterization of languages which are not in  $\Wh(\preceq)$ was  based on the existence of particular cycles in the product automaton $\mt D_\La^2$.
Here we shall use a variant of this result  based on the following definition:

\begin{definition}\label{def: intr}
Let $\mt D= (Q, \Sigma, \delta, s, F)$ be a DFA and  let  $\preceq$ be an order of the alphabet. If    $p,q\in Q$  we define: 
\begin{align*}
    p \intro q &\iff  \exists  \alpha, \alpha' \in I_p,~  \exists\beta, \beta' \in I_q ~ (\alpha \prec  \beta)\wedge  (\beta' \prec   \alpha'). \\
    p \intr q &\iff  \exists \preceq ~p \intro  q.
\end{align*}
\end{definition}

  Recall that  the automaton $\mt D_\La^2$ is defined by $\mt D_\La^2=(Q\times Q, \Sigma, \delta', (s,s), F\times F)$ where $\delta'((p,q),a)=(\delta(p,a), \delta(q,a))$.  Then we have: 
  
\begin{lemma}[{\cite[Theorem 3]{becker2023optimal}}]\label{lem:wcycle}
  Let $\La$ be a regular language,   let $\preceq$ be an order of the alphabet, and let  $\mt D_\La$ be the minimum trimmed DFA accepting $\La$. Then, 
    $\La\notin \Wh(\preceq)$ if and only if $\mt D_\La^2$ contains a cycle $(p_1, q_1) \rightarrow
    (p_2, q_2) \rightarrow \cdots \rightarrow (p_k, q_k) \rightarrow (p_1, q_1)$ such that
 the following hold: (i) $p_1 \neq q_1$, and (ii) $p_1 \intro  q_1$.
 \end{lemma}

It is not difficult to prove that the class $\Wh(\preceq)$ is closed under intersection. 
However, it is neither closed under complement nor union: indeed, there are examples of languages $\La\in \Wh(\preceq)$ for which no  \emph{complete} WDFA recognizing $\La$ exists (see \cite{AlankoDPP21}) and this prevent the possibility  of finding a WDFA for the complement by simply switching final and non-final states~\cite{Castiglione}.

One peculiarity of the Wheeler notion is that it depends on a fixed order of the alphabet.
In order to compare Wheeler languages with other sub-classes of  regular languages, which do not depend on a fixed order of the alphabet, it is natural to ask for a definition that is independent of such an order. Two natural options are the following. 

\begin{definition}
A regular language $\La$ is:
\begin{itemize}
\item [-]   {\em Existentially Wheeler} if there exists an order $\preceq$ such that $\La\in \Wh(\preceq)$;
\item [-]   {\em Universally Wheeler} if $\La\in \Wh(\preceq)$ for all orders $\preceq$  of the alphabet.
  \end{itemize}
The class of Existentially (Universally) Wheeler languages is denoted by $\EW$ ($\UW$).
\end{definition}

In \cite{DBLP:journals/tcs/DAgostinoMP23} it is proved that deciding \EW{} is  $\NP$-complete.  In Section~\ref{sec:decidingUW} we shall prove that deciding \UW\ is in $\P$, but first in the next lemma we give a necessary condition that languages in \EW{} must satisfy. 

\begin{lemma}\label{lem:3cycles}
    If the minimum trimmed automaton $\mt D_\La$ of a language $\La$ contains three equally labeled cycles starting from three different states, then $\La\not \in \EW$.
\end{lemma}

Note that the condition of Lemma~\ref{lem:3cycles} can easily be checked in cubic time. 

\section{Definite and Reverse Definite Languages meet  Wheelerness} 
\label{sec:DefRevDef}

We first notice  that the class $\UW$ is not closed under complement (see Example \ref{ex} in the appendix). 
In this section we give a characterization of the class of languages $\La$ for which both $\La$ and $\overline \La$ are in the class $\UW$. This characterization involves two well known subclasses of star-free languages, the definite and reverse definite ones \cite{Brzozowski}. Remember that, given two languages $\La, \La'$,  we denote by $\La\La'$ the set containing all concatenations of a word in $\La$ with a word in $\La'$ (with $\La\La'$ equal to the empty language if either one of $\La, \La'$ is). 

\begin{definition}\;A regular language $\La$ is:
 \emph{definite} (\DEF{}) if  $\La=F\cup \Sigma^*G$, for finite  $F,G\subseteq \Sigma^*$;
  \emph{reverse definite} (\RDEF{}) if  $\La=F\cup G\Sigma^* $, for finite  $F,G\subseteq \Sigma^*$;
 \emph{strictly locally testable} (\SLT{}, see~\cite{Caron}) if  there are finite sets $H,K,W,F$ s.t.
  $\La=F\cup (H\Sigma^* \cap \Sigma^* K) \setminus \Sigma^* W \Sigma^*$.
\end{definition}

Our aim is to prove that the class $\DEF\cup \RDEF$ consists in exactly the $\UW$-languages for which the complement is in $\UW$ as well. 
First, using a characterization due to Caron~\cite{Caron}, we compare $\SLT$ and $\UW$ languages.

\begin{lemma} [{\cite[Proposition 3.2 and 3.3]{Caron}}] \label{lem:caron} Let $\La$ be a regular language and let $\mt D_\La$ be its minimum trimmed DFA. Then  $\La\not \in \SLT$ if and only if there exists  a pair of equally labeled cycles $\mt C_p, \mt C_q$ in $\mt D_\La$ starting respectively from  states $p\neq q$.  
    \end{lemma}

   Since a  pair of equally labeled cycles in $\mt D_\La$ corresponds to a single  cycle  in $\mt D_\La^2$, combining the previous  lemma with Lemma~\ref{lem:wcycle} we obtain:

    \begin{lemma}\label{lem:slt_are_uw} $\SLT \subseteq \UW$. In particular, all languages in $\DEF\cup \RDEF$ are in $\UW$. 
    \end{lemma}

In order to prove that the class of regular languages $\La$ for which both  $\La$ and $\overline \La$ are in $\UW$ coincides with the class $\DEF\cup \RDEF$     we need a careful inspection,  for a regular language $\La$,  on   the differences between  the minimum complete DFA for $\La$,  $D^c_\La$, its trimmed version $D_\La$,    the minimum complete DFA for $\overline \La$, $  D^c_{\overline {\mt  L}}$, and its trimmed version  
  $ D_{\overline {\mt  L}}$, as  in the following remark. 

  \begin{remark}\label{rem:careful}
Since the Myhill-Nerode equivalence on $\Sigma^*$ is the same for $\La$ and $\overline{\La}$,  we have that $  D^c_\La$ and $  D^c_{\overline {\mt  L}}$ only differ for the final states,  which  in  $  D^c_{\overline {\mt  L}}$ is the  complement of the one in $  D^c_\La$. 
Let $Q$ be the set of $  D^c_\La$-states.  
Consider the subsets $Q_1,Q_2,Q_3\subseteq Q$ 
 defined as follows:
 \begin{eqnarray*}
    Q_1 & = & \{q \in Q \mid I_q \cap (\pf\La \setminus \pf{\overline\La})\neq \emptyset\}, \\
     Q_2 & = & \{q \in Q \mid I_q \cap (\pf{\overline\La} \setminus \pf\La)\neq \emptyset\}, \\
     Q_3 & = & \{q \in Q \mid I_q \cap (\pf{\overline\La} \cap \pf\La)\neq \emptyset\}.
 \end{eqnarray*}
States in \(Q_1\) are final for \(\mt D^c_\La\) and can reach only final states, hence, by minimality,  \(|Q_1|\leq 1\). Moreover,   \(|Q_1|= 1\) iff \(\pf\La \setminus \pf{\overline\La}\neq \emptyset\), and if \(Q_1= \{q\}\)  then $q$ is an  absorbing state  (that is, $\delta(q,a)=q$, for all $a\in \Sigma$)  and it    is final for \(\mt D^c_{{\La}}\). 
The analogous  holds for \(Q_2\) with respect to \(\mt D^c_{\overline{\La}}\):  \(|Q_2|= 1\) iff \(\pf{\overline\La} \setminus \pf\La\neq \emptyset\), and if \(Q_2= \{\overline q\}\)  then $\overline q$ is an  absorbing state which is final in 
\(\mt D^c_{\overline{\La}}\).
%States in \(Q_3\) are states for  both  $\mt D^c_\La$   and  $\mt D^c_{\overline \La}$.
Finally,    $D_\La$ is obtained from $D^c_\La$ by (eventually) erasing the   absorbing state $\overline q$ and all transition arriving in $\overline q$, while   $ D_{\overline {\mt  L}}$  is obtained from $ D^c_{\overline {\mt  L}}$ by (eventually) erasing the  absorbing state $ q$ and all transition arriving in $q$ (see Example~\ref{ap:ex} in the appendix). 
\end{remark}

Considering  a  regular language $\La$ with $\pf\La=\Sigma^*$ we notice the following:  
\begin{remark} \label{rem:complement}  For every order $(\Sigma,\preceq)$, if $\pf\La=\Sigma^*$  and    
$\La\in \Wh(\preceq)$ then $\overline \La\in \Wh(\preceq)$.   
 This follows by the fact that if $\pf\La=\Sigma^*$, then any deterministic
 automaton accepting $\La$ is complete and, thus, a  Wheeler DFA for $\La$ becomes a Wheeler DFA for $\overline \La$ by just exchanging final and not final states.  In particular, if $\La\in \UW$ and $\pf\La=\Sigma^*$,  then $\overline \La\in \UW$. Notice that  the converse is not true, because  there are example in which  $\pf\La=\Sigma^*$,  $\overline \La\in \UW$ but  $\La\not \in \UW$ (see 
 Example~\ref{ex} in the appendix).
 \end{remark}

Using the previous remarks we can characterize definite languages using $\UW$.
\begin{lemma}\label{lem:Sigma^*}   
  $\La\in \DEF$   iff  $\La$ is finite,  or  $\La\in \UW$ and  $\pf\La=\Sigma^*$.   
\end{lemma}

We now consider the case $\pf\La\neq \Sigma^*$. 
If  $|\Sigma|=1$ then  $\UW=\EW$  coincide with the class of finite and cofinite languages which, for unary alphabet, is    equal to $\DEF\cup \RDEF$. 
Hence, for unary languages,  $\La, \overline \La\in \UW$   iff    $\La\in \DEF\cup \RDEF$ holds.  
To handle the case $|\Sigma|\geq 2$ we prove:

\begin{lemma}\label{lemma:esiste}
Let  $|\Sigma|\geq 2$ and let  $\La\in \UW$. If $\pf\La\neq \Sigma^*$, then
$\exists \preceq \;\overline \La\not \in \Wh(\preceq)$ iff $|\pf\La\cap \pf{\overline\La}|=\infty$.

\end{lemma}

  Using Lemma~\ref{lemma:esiste} we can finally prove:

\begin{corollary} \label{cor:nonSigma^*}  Let  $|\Sigma|\geq 2$ and let  $\La$ be a regular language such that   
  $\pf\La\neq \Sigma^*$. 
Then  $\La, \overline \La\in \UW$   iff    $\La\in \RDEF$  .
\end{corollary}
\begin{proof} 
If $\La\in \RDEF$ then  $\La, \overline \La \in \SLT$ and hence in $\La, \overline \La \in \UW$ by Lemma~\ref{lem:slt_are_uw}.

Suppose that $\La, \overline \La\in \UW$ and 
    $\pf\La\neq \Sigma^*$.  Then  from Lemma~\ref{lemma:esiste}  we obtain that    $\pf\La\cap \pf{\overline\La}=F$ where  $F$  is a finite set. \\
If   $G= \{\alpha x \in \pf\La\setminus \pf{\overline\La}~:~  \alpha \in F,~ x\in \Sigma \}$,   we claim that   
 $\La=(F\cap \La)\cup G\Sigma^*$. 
If $\beta\in G\Sigma^*$ then  $\beta=\alpha x \gamma$ with  $\alpha x\in \pf\La\setminus \pf{\overline\La}$  and $\alpha \in F$. Since  $\alpha x\not \in   \pf{\overline\La} $, 
for every  $\delta \in \Sigma^*$ we have  $\alpha x \delta\in \La$. In particular, 
 $\beta=\alpha x \gamma\in \La$, proving  $(F\cap \La)\cup G\Sigma^*\subseteq \La$. 

Conversely, we prove that if   $\beta\in \La$, then    $\beta \in (F\cap \La)\cup G\Sigma^*$. 
 If  $\beta\in \pf{\overline\La}$ then  $\beta \in \pf\La\cap \pf{\overline\La}= F$.
Hence  $\beta\in F\cap \La$.
 
If  $\beta\not \in  \pf{\overline\La}$, let  $\alpha$ be  its longest prefix belonging to      $\pf\La\cap   \pf{\overline\La}=F$.  Since  $\alpha\neq \beta$, there exists  $x\in \Sigma$ such that  $\alpha x$ is still a   $\beta$-prefix  and hence  $\alpha x \in \pf\La$.  From  the definition of $\alpha$ it follows that   $\alpha x \in \pf\La\setminus    \pf{\overline\La}$. Hence $\alpha x \in G$ and $\beta \in G\Sigma^*$.

Since we proved that   $\La=(F\cap L)\cup G\Sigma^*$ with  $G,F$ finite sets, 
we obtain $\La\in \RDEF$. 
\end{proof}

Putting all previous results together we finally obtain:

\begin{theorem}\label{th:antonio}  Let $\La$ be a regular language. Then
\(\La, \overline \La\in \UW\Leftrightarrow \La\in \DEF\cup \RDEF. \)
    \end{theorem}
    \begin{proof} Since the class $\DEF\cup \RDEF$ is closed for complements,  implication from right to left is Lemma~\ref{lem:slt_are_uw}. 
    The implication from left to right  is proved just before Lemma \ref{lemma:esiste} for  $|\Sigma|=1$ 
   
    and,  if  
    $\pf\La=\Sigma^*$ then the result  follows from Lemma~\ref{lem:Sigma^*};    
    if $\pf\La\neq \Sigma^*$ and $|\Sigma|\geq 2$ and  then the result  follows from Corollary~\ref{cor:nonSigma^*}.
    \end{proof}

\section{Deciding $\UW$}\label{sec:decidingUW}

In this section we consider the  problem of deciding, given a DFA $\mt D$,   whether $\La(\mt D) \in \UW$. First, we propose an algorithm that solves  this problem in quadratic time, more precisely in time $\mt O(nm)$, where $n$ and $m$ are respectively the number of states and edges of the minimal DFA $\mt D_{\La}$ accepting the language $\La$ (note that $n$ and $m$ are upper bounded by the number of states and edges of the input DFA $\mt D$). Finally, we prove that this is optimal unless the Strong Exponential Time Hypothesis (\SETH{})~\cite{impagliazzo2001complexity,vassilevska2015hardness} fails.

\subsubsection{The algorithm}\label{subsec:algorithm}

A characterization of $\UW$ can be given by looking at the structure of the minimum trimmed automaton. Indeed, since the existence of a cycle in $\mt D_\La^2$ (see Section~\ref{sec:prelim}) does not depend on the alphabet's order, from Lemma~\ref{lem:wcycle} we obtain: 
\begin{corollary}\label{cor:strategy}
   Let $\mt D_\La$ be the minimum  trimmed  DFA accepting $\La$. Then, $\La\notin \UW$ if and only if  $\mt D_\La^2$ contains a cycle $(p_1, q_1) \rightarrow
   (p_2, q_2) \rightarrow \cdots \rightarrow (p_k, q_k) \rightarrow (p_1, q_1)$ such that
   the following hold: (i) $p_1 \neq q_1$, and (ii) $p_1 \intr  q_1$.
\end{corollary}
Note that this corollary provides us with a strategy to decide, for a given regular language $\La$, whether  $\La\in \UW$ without testing every possible order $\preceq$ of the alphabet. It thus motivates the following definition of a directed (unlabeled) graph. Here and in the rest of this section $Q,\delta$ always refer to the set of states and the transition function of  $\mt D_\La$, the minimum trimmed DFA for $\La$. 
\begin{definition}\label{def:graph V E}  Let  $(V,E)$ be such that $   V:=\Set{(p,q)\in Q^2}{(p\neq q) \land  (p \intr q)}$ and   
$E:=\Set{\langle (p,q), (p',q') \rangle}{\exists a\in \Sigma~\bigl( \delta(p,a)=p' \wedge \delta(q,a)=q'\bigr) }$.
\end{definition}
Using the graph $(V,E)$, we can rephrase Corollary~\ref{cor:strategy}    as follows.
\begin{corollary}\label{cor:testingUW}
    $\La  \in \UW$ if and only if the graph $G=(V,E)$ of Definition~\ref{def:graph V E} is acyclic.
\end{corollary}

\paragraph{More Tractable Definition of $V$.}
Given the above corollary, in order to decide $\UW$ it is sufficient to test  the directed graph $(V,E)$ for acyclicity. However, the problem is that the above definition of $V$ does not easily allow us to efficiently compute the graph $(V, E)$.
To this end, we show how  to  define $V$ in an equivalent way that enables us to build  the graph $(V,E)$  in time $ \mathcal{O}(nm)$, starting from the DFA  ${D_\La}$ having   $n$   states and $m$  edges. 
We will now define a table $P$ indexed by pairs   $(p,q)\in Q^2$, such that  $P[p, q]$    collects   all pairs $(a,b)\in \Sigma^2$ with $a\neq b$  which are   witnesses  of the existence   of words of the form
$\alpha=\alpha' a \gamma, \beta=\beta'b\gamma$ arriving in $p,q$, respectively.   Together with pairs $(\dashv, \dashv)$ (or $(\vdash, \vdash)$) indicating whether there are words $\alpha\dashv \beta$   ($\beta\dashv\alpha$) arriving in $p,q$, respectively, these pairs are enough to decide whether $p\intr q $.  Recall that $\alpha \dashv \beta$ means that the word $\alpha$ is a suffix of $\beta$. We extend this definition to states as follows:

\begin{definition} 
    Let $p, q \in Q$ with $p \neq q$. Let 
   \( p \dashv q   \iff  \exists \alpha \in I_p  \exists \beta \in I_q (\alpha \dashv \beta)\). Furthermore, 
\begin{align*}
    P[p, q]   &\coloneqq  
    \Set{ (a, b) \in \Sigma^2 }{ a \neq b \land  \exists \alpha, \beta, \gamma \in \Sigma^* 
           (\alpha a \gamma \in I_p \land \beta b \gamma \in I_q) } ~\cup\\
                   &\;\;\;\;\;\;\Set{(\dashv, \dashv)}{p \dashv q} ~ \cup~ \Set{(\vdash, \vdash)}{q \dashv p}.
\end{align*}
\end{definition}

Notice that, since   $\mt  D_\La$ is trimmed  we have, for  $(a,b)\in \Sigma^2$ with $a \neq b$:
\begin{align}\label{formula: equiv P}
\begin{split}
    &(a,b) \in\ P [p,q]  \Longleftrightarrow \\ 
    &\exists p', q' \in Q,~\exists \gamma \in \Sigma^*~ 
    (a \in \lambda(p') \land
               b \in \lambda(q') \land p = \delta(p', \gamma) \land q = \delta(q', \gamma)).
\end{split}
\end{align}
 
\begin{lemma} \label{lem:Pgeq2}
    Let $p, q \in Q$ with $p \neq q$. Then, $p \intr q$ if and only if $|P[p, q]| \geq 2$.
\end{lemma}

Hence, in order to decide whether  $p\intr q$ it will be sufficient to compute just two elements (if they exist)
from $P[p, q]$. 
Since we are only interested in deciding whether $P[p, q]$ contains at least two elements, it seems plausible
that it is sufficient to examine only a few labels of $\lambda(q)$ for each $q\in Q$. Indeed, in this spirit suppose we have computed a maximal set $\lambda'(p) \subseteq \lambda(p)$ such that 
$|\lambda'(p)| \leq 2$; in other words, $\lambda'\colon Q\rightarrow 2^\Sigma$ is a function such that 
$\lambda'(p) = 
\lambda(p)$, if  $|\lambda(p)|\leq 1$,  while $\lambda'(p)$ is a subset of $\lambda(p)$ with two elements, if $|\lambda(p)|> 1$. For example, we may define  $\lambda'(p)$ as the set containing the two smallest elements in $\lambda(p)$ (or
$\lambda(p)$ itself, in case there are not enough elements). Given such a function $\lambda'$, we define:
\begin{align*}\label{def:lambda'}
   P_{\lambda'}[p, q] \coloneqq 
    &\{ (a, b) \,|\, a \neq b \, \land \,\exists p', q' \in Q, ~ \exists \gamma \in \Sigma^*~\\
    &\hspace{3em}  (a \in \lambda'(p') \land
                b \in \lambda'(q') \land p = \delta(p', \gamma) \land q = \delta(q', \gamma)) \} \\
                  \cup &\Set{(\dashv, \dashv)}{p \dashv q}
                  \cup \Set{(\vdash, \vdash)}{q \dashv p}.
\end{align*}
We note that due to the characterization of $P[p,q]$ in~\eqref{formula: equiv P}, the definition of $P_{\lambda'}[p, q]$ is equivalent to the one of $P[p, q]$ up to replacing $\lambda$ by $\lambda'$. Moreover, since $\lambda'(\cdot) \subseteq \lambda(\cdot)$, it holds that $P_{\lambda'}[p, q] \subseteq P[p, q]$
for every $p, q \in Q$.

The following lemma ensures that we can use $\lambda'(\cdot)$ instead of $\lambda(\cdot)$.

\begin{lemma}\label{lemma:P-table-geq2}
    Let $p, q \in Q$, with $p \neq q$. Then $|P[p, q]| \geq 2$ iff $|P_{\lambda'}[p, q]| \geq 2$. 
\end{lemma}

\paragraph{Algorithm.}
We are now ready to describe the algorithm that tests   whether a language $\La$ is in $\UW$ using Corollary~\ref{cor:testingUW}.

We start by constructing the vertices $V$ of the graph $(V, E)$ as follows. We compute sets $\lambda'(p)$ for each state $p$ in $\mt D_\La$ as above, i.e., $\lambda'(p)$ is maximal with $\lambda'(p) \subseteq \lambda(p)$ and $|\lambda'(p)| \leq 2$. Then according to Lemma~\ref{lem:Pgeq2} and Lemma~\ref{lemma:P-table-geq2}, it follows that the vertex set $V$ can be equivalently defined as $V=\{(p,q)~:~ P_{\lambda'}[p,q]|\geq 2\}$. The pseudocode is given in Algorithm~\ref{algo:T-table-new}.

First, in lines~\ref{line: lambda prime start} to~\ref{line: lambda prime end} we compute the   sets  $\lambda'(q)$, for all $q\in Q$. Then, from line~\ref{line: A1 start} to~\ref{line: A1 end}    the algorithm  computes a queue $A$ containing all quadruples $(s,q, \dashv, \dashv)$,  $(p,s, \vdash, \vdash)$,  where $s$ is the initial state of   $\mt  D_\La$ and $s\neq p,  q$.  Notice that, since $\mt  D_\La$ is a trimmed automaton,      both    $s\dashv q$ and $p\vdash s$  hold so that  $(\vdash, \vdash)\in P_{\lambda'}[s,q]$ and   $(\dashv, \dashv)\in P_{\lambda'}[p,s]$, for all $s\neq p,   q$.  Then, from line~\ref{line: A2 start} to~\ref{line: A2 end},    all quadruples $(p, q, a, b)$ such that $p\neq q$, $a\neq b$, $a\in \lambda'(p)$,  $b\in \lambda'(q)$ are added to $A$. 
These  quadruples represent the ``base'' case  in which  $(a,b)\in P_{\lambda'}[p,q]$, where,  in  the definition of $P_{\lambda'}[p,q]$  we consider   $\gamma=\epsilon$, $p'=p$, and  $q'=q$.  

Starting from these quadruples, from line~\ref{line: propagation start} to~\ref{line: propagation end}   we  ``propagate'', generating new pairs $(a,b)\in P_{\lambda'}[p,q]$, as long as  $P_{\lambda'}[p,q]<2$. This is correct, because, by definition, for all $(p',q')$ we have $(a,b)\in P_{\lambda'}[p',q']$ iff there exists $\gamma\in \Sigma^*$ and $(p,q)$ such that $(p, q, a, b)$ belongs to $A$ and $\delta(p,\gamma)=p', \delta(q,\gamma)=q'$.
At the same time we compute a table $T$ such that,   for each $p\neq q$, $T[p,q]$ contains   two pairs in   $P_{\lambda'}[p,q]$ if there are such.  

Finally, from line~\ref{line: build graph start} to~\ref{line: build graph end}, we select the vertices $V$ using $T$ (i.e., those pairs $(p,q)$ with $|T[p,q]| \geq 2$), build the edge set $E$, and test acyclicity of the subgraph $G=(V,E)$ of $\mt D_\La^2$. 

\begin{algorithm}[ht!]
\caption{Test if $\La(D_\La) \in \UW$.}\label{algo:T-table-new}
\begin{algorithmic}[1]
    \Function{$\UW$-Tester}{$\mt D_\La= (Q, \Sigma, \delta, F)$}
        \ForAll{$(p, a, q) \in \delta$}\Comment{Compute $\lambda'(q)$ for every $q \in Q$.} \label{line: lambda prime start}
            \If{$|\lambda'(q)| < 2 \land a \notin \lambda'(q)$}
                \State $\lambda'(q).\Call{Push}{a}$
            \EndIf
        \EndFor\label{line: lambda prime end}
        \State
        \State $A \gets \Call{EmptyQueue}{}$\Comment{Base cases.} \label{line: A1 start}
        \ForAll{$p, q \in Q$ s.t. $p \neq q$}
            \If{$p = s$}
                \State $A.\Call{Push}{(s, q, \dashv, \dashv)}$
            \EndIf
            \If{$q = s$}
                \State $A.\Call{Push}{(p, s, \vdash, \vdash)}$
            \EndIf \label{line: A1 end}
            \ForAll{$(a, b) \in \lambda'(p) \times \lambda'(q)$} \label{line: A2 start}
                \If{$a \neq b$}
                    \State $A.\Call{Push}{(p, q, a, b)}$
                \EndIf\label{line: A2 end}
            \EndFor
        \EndFor
        \State
        \State $T \gets \Call{Empty}{Q \times Q}$\Comment{Propagation phase.} \label{line: propagation start}
        \While{$A \neq \emptyset$}
            \State $(p, q, a, b) \gets A.\Call{Pop}{ }$
            \If{$|T[p, q]| < 2 \land (a, b) \notin T[p, q]$}
                \State $T[p, q].\Call{Push}{(a, b)}$
                \ForAll{$(p', q')$ outgoing from $(p, q)$ in $\mt  D^2$}
                    \If{$p' \neq q'$}
                        \State $A.\Call{Push}{(p', q', a, b)}$
                    \EndIf
                \EndFor
            \EndIf
        \EndWhile \label{line: propagation end}
        \State
        \State $V \gets \Set{(p,q) \in Q\times Q}{|T[p,q]| \geq 2}$ \Comment{Subgraph of $\mt  D^2$ of interest.} \label{line: build graph start}
        \State $E \gets \Set{(u, v) \in V\times V}{\text{$v$ is outgoing from $u$ in $\mt  D^2$} }$ \label{line: build graph end}
        
        \State\Return $\Call{IsAcyclic}{V, E}$
    \EndFunction
\end{algorithmic}
\end{algorithm}

\paragraph{Running Time.}
For an automaton $\mt  D_\mathcal{L}$ with $n$ states and $m$ transitions whose squared automaton
$\mt  D_\mathcal{L}^2$ has $n^2$ states and at most $nm$ transitions, the time required by the respective steps of the above algorithm is as follows:
\begin{enumerate}
    \item The computation of $\lambda'$ requires $\mathcal{O}(m)$ time,
    \item Computing the table $T$ can be done in $\mathcal{O}(nm)$ time,
    \item  Computing the graph $G = (V, E)$ takes $\mathcal{O}(nm)$ time,
    \item Finally, testing acyclicity of $G$ takes $\mathcal{O}(nm)$ as well.
\end{enumerate}
Altogether, Algorithm~\ref{algo:T-table-new} runs in $\mathcal{O}(n^2 + nm) \subseteq \mathcal{O}(nm)$ time. 

\subsubsection{Conditional Lower Bound.}\label{subsec:optimal}
First, we define the Orthogonal Vectors problem.
\begin{definition}[Orthogonal Vectors problem (OV)]\label{def:OV}
    Given two sets $A$ and $B$, each containing $N$ vectors from $\{0,1\}^d$, decide whether there exist $a\in A$ and $b\in B$ such that $a^{\intercal}b = 0$. 
\end{definition}

Through a classical reduction~\cite{Williams05}, for $d\in\omega(\log N)$, it is known that \OV{} cannot be solved in time $\mt O(N^{2-\eta} \poly(d))$ for any constant $\eta>0$
unless the Strong Exponential Time Hypothesis \cite{impagliazzo2001complexity,vassilevska2015hardness} (\SETH{}) fails. Our goal now is to prove the following theorem.

\begin{theorem}\label{thm:lower bound}
    Let $\La$ be a regular language, let $\mt  D_\La$ be the minimum DFA accepting $\La$ and let $m$ be the number of transitions in $\mt  D_\La$. 
    If the problem of checking whether $\La\in \UW$ can be solved in time $\mt O(m^{2-\eta})$ for some constant $\eta>0$, then the Orthogonal Vectors problem with $d \in \Omega(\log N)$ can be solved in time $\mt O(N^{2-\eta} \poly(d))$, thus contradicting \SETH{}.
\end{theorem}
We prove this theorem using the following lemma, which reduces an instance of the \OV{} problem with two sets of $N$ vectors in $d$ dimensions each into a minimum DFA (and thus a regular language $\La$ for which we can check if $\La\in \UW$) of size $\Theta(Nd)$.
\begin{lemma}\label{lemma: reduction}
    For an instance of \OV{}, we can in $\mt O(N(d+\log N))$ time construct a DFA $\mt D$ 
    with $m = \Theta(N(d+\log N))$ edges
    that is minimum for its language such that the \OV{} instance is a \YES-instance if and only if $\mt  D$ contains two distinct nodes $u,v$ with $\delta(u, \gamma) = u$ and $\delta(v, \gamma) = v$ for some string $\gamma$ and $u\intr v$.
\end{lemma}
Once this lemma is established, we can take an \OV{} instance with sets of size $N$ containing vectors of dimension $d \in \omega(\log N)$ and construct the DFA $\mt D$ of size $m = \Theta(N(d+\log N)) = \Theta(Nd)$. Now assume that we can check if $\La\in \UW$ for the language recognized by $ \mt D$ in $\mt O(m^{2-\eta})$. Using Corollary~\ref{cor:strategy} and Lemma~\ref{lemma: reduction}, we can thus solve the \OV{} instance in $\mt O((Nd)^{2-\eta}) = \mt O(N^{2-\eta}\poly(d))$ time, as the \OV{} instance is a \YES-instance if and only if $\La\notin \UW$. This shows Theorem~\ref{thm:lower bound}.

In the rest of this section we argue why Lemma~\ref{lemma: reduction} holds. Our reduction is indeed identical to the one of Becker et al.~\cite{becker2023optimal}, however we need a few additional observations in order to adjust the proof of their Proposition~1 to a proof of Lemma~\ref{lemma: reduction}.

The construction given in Proposition 1 of \cite{becker2023optimal} gives a minimum DFA $\mt  D$ over the alphabet $\{\#,0,1\}$ from an instance $A=\{a_1,...,a_N\}, B=\{b_1, ...,b_N\}$ of \OV{} such that:
\begin{enumerate}
\item $\mt  D$ contains pairwise distinct nodes  $\hat{a}_1,...,\hat{a}_N, \hat{b}_1, ...,\hat{b}_N$, such that, for all states  $u\neq v$ in $\mt  D$, there exists $\gamma$ with  $\delta(u, \gamma) = u$ and $\delta(v, \gamma) = v$ for some string $\gamma$ iff $u= \hat{a}_i, v= \hat{b}_j$, and $a_i^{\intercal}b_j=0$ holds. 
\item every node $\hat{a}_i$ is reached by a word  ending in $000$ and by a word ending in $110$; every node $\hat{b}_j$ is reached by a word  ending in $010$.
\end{enumerate}

These properties  are sufficient  to prove Lemma~\ref{lemma: reduction}.
On the one hand,  if $A=\{a_1,...,a_N\}, B=\{b_1, ...,b_N\}$ are such that    $a_i^{\intercal}b_j=0$ for some $i,j$,  then,   by 1  above,  in  $\mt  D$  we can find $u=\hat{a}_i, v=\hat{b}_j$  and a string  $\gamma$ with  $\delta(u, \gamma) = u$ and $\delta(v, \gamma) = v$. Moreover,   by 2  above, there are $\alpha, \alpha'\in I_u$ and $\beta\in I_v$ such that $000\dashv \alpha, 110\dashv \alpha'$ and $010\dashv \beta$. Then $u\intr v$ holds, because     we have $\alpha\prec \beta$ iff $\beta\prec \alpha'$,  for any order  $\preceq$ of the alphabet $\{\#,0,1\}$.

On the other hand,  if $\mt  D$ contains two distinct nodes $u,v$ with $\delta(u, \gamma) = u$ and $\delta(v, \gamma) = v$ for some string $\gamma$ and $u\intr v$, then, by 1 above,   there are $i,j$ such that  $u= \hat{a}_i, v= \hat{b}_j$ and $a_i^{\intercal}b_j=0$ holds.

\section{Conclusions}
In this paper, we established a link between the class of Wheeler languages and --- well studied --- (sub-)classes of \emph{locally testable} languages. Our link is provided on both theoretical as well as algorithmic grounds. On the one hand, we proved that the class of languages that, together with their complement, are Wheeler with respect to \emph{any} order of the alphabet corresponds exactly to the union of Definite and Reverse Definite languages --- that is, the collections of words characterized either by a finite set of prefixes or suffixes. On the other hand we provided an optimal algorithm for characterizing the \emph{universally} (with respect to the alphabet ordering) Wheeler languages. 

The next steps along the path we started to follow  here are, for example, to characterize languages that, together with their complement, are Wheeler with respect to a fixed order of the alphabet. Moreover, somehow along the style proposed in \cite{jacm/CotumaccioDPP23}, it would be interesting to reformulate and study analogues of the problems studied here for general --- i.e. co-lexicographically \emph{partially} ordered --- languages/automata.

\begin{subappendices}
\renewcommand{\thesection}{\Alph{section}}%

\section{Proofs from Section \ref{sec:DefRevDef}}

The proof of Lemma \ref{lem:wcycle} relies on the following results from \cite{AlankoDPP21}:
\begin{theorem}[{\cite[Theorem 4.3]{AlankoDPP21}}]\label{th:Wheeler} Let $\La$ be a regular language   and let $\mt D_\La$ be its minimum trimmed DFA. Then  $\La\not \in \Wh(\preceq)$ if and only if all of the following hold:
\begin{enumerate}
    \item there exists a pair of equally-labeled   cycles $\mt C_p, \mt C_q$ (say labeled by $\gamma \in \Sigma^+$) in $\mt D_\La$  starting respectively from  states $p\neq q$;
    \item there exist two words $\alpha, \beta$ with $\alpha\in I_p, \beta\in I_q $ such that $\gamma$ is not a suffix of $\alpha $ and  $\gamma$ is not a suffix of $\beta$;
\item   either $\alpha, \beta \prec \gamma$ or $\gamma \prec \alpha, \beta$. 
\end{enumerate}
\end{theorem}

\begin{lemma*}[\ref{lem:wcycle}]  
  Let $\La$ be a regular language,   let $\preceq$ be an order of the alphabet, and let  $\mt D_\La$ be the minimum trimmed DFA accepting $\La$. Then, 
    $\La\notin \Wh(\preceq)$ if and only if $\mt D_\La^2$ contains a cycle $(p_1, q_1) \rightarrow
    (p_2, q_2) \rightarrow \cdots \rightarrow (p_k, q_k) \rightarrow (p_1, q_1)$ such that
 the following hold: (i) $p_1 \neq q_1$, and (ii) $p_1 \intro  q_1$.
 \end{lemma*}

 \begin{proof}
  Suppose $\La$ is not  Wheeler w.r.t. $\preceq$. Then, by Theorem \ref{th:Wheeler} there are:  
\begin{enumerate}
    \item a  pair of equally labeled   cycles $\mt C_p, \mt C_q$ (say labeled by $\gamma$) in $\mt D_\La$  starting starting respectively from  states $p\neq q$;
    \item two words $\alpha, \beta$ with $\alpha\in I_p, \beta\in I_q $ such that $\gamma$ is not a suffix of $\alpha $ or $\beta$;
\item   either $\alpha, \beta \prec \gamma$ or $\gamma \prec \alpha, \beta$. 
\end{enumerate} Let $p_1=p,q_1=q$ and consider the $\gamma$-cycle       $(p_1, q_1) \rightarrow
    (p_2, q_2) \rightarrow \cdots \rightarrow (p_k, q_k) \rightarrow (p_1, q_1)$  in $\mt D_\La^2$.
Then it is   sufficient to prove that  $p_1\intro q_1$. 
Suppose w.l.o.g. that $\alpha\prec \beta$. If 
  $\alpha, \beta \prec \gamma$ then $\alpha\gamma, \alpha \gamma^2\in I_{p_1}$, $\beta \gamma \in I_{q_1}$ and $\alpha\gamma \prec\beta\gamma$ while  $\beta\gamma\prec \alpha \gamma^2$.   If $\gamma \prec \alpha, \beta$ then 
  $\beta\gamma, \beta\gamma^2\in I_{q_1}$, $\alpha\gamma \in I_{p_1}$ and 
  $\beta\gamma\succ \alpha\gamma$ and $\alpha\gamma  \succ \beta\gamma^2$ (the last disequality holds because $\gamma $  is not a suffix of $\alpha$).
Hence in both cases we obtain   $p_1\intro q_1$. 

\medskip

Suppose now that we have two states $p_1\neq q_1$ with  $p_1\intro q_1$  in $\mt D_\La$  from which two   cycles labeled by the same word $\gamma$ start.  From  $p_1\intro q_1$ we know that there are $\alpha, \alpha'\in I_{p_1}$, $\beta, \beta'\in I_{q_1}$ with $\alpha \prec \beta $ and $\beta'\prec \alpha'$. By using a power of $\gamma$, if necessary, we may suppose that $\gamma$ is not a suffix of $\alpha, \beta, \alpha', \beta'$. If $\gamma \prec \alpha$  or $\alpha'\prec \gamma$ then the conditions of Theorem \ref{th:Wheeler}    are satisfied (with $\alpha', \beta'$ playing the role of $\alpha, \beta$ in the second case).
Otherwise, we have $\alpha \prec \alpha'$. Then, either  $\beta\prec \gamma$ or   $\gamma \prec \beta$ holds. In the first case we have $\alpha,  \beta \prec \gamma$ while in the second case we have $\gamma \prec \beta, \alpha'$, so that, in both cases we fulfill the  conditions of Theorem ~\ref{th:Wheeler}. 
\end{proof}

\begin{lemma*}[\ref{lem:3cycles}]
    If the minimum trimmed automaton $\mt D_\La$ of a language $\La$ contains three equally labeled cycles starting from three different states, then $\La\not \in \EW$.
\end{lemma*}
\begin{proof} Suppose $\mt D_\La$ contains three different states $p,q,r$ which are the starting point  of three   cycles labeled by the same word $\gamma$. Let $\zeta_1, \zeta_2,\zeta_3$ be words reaching $p,q,r$ from the starting state. Using a power of $\gamma$, if necessary, we may suppose that $\gamma$ is not a suffix of   $\zeta_i$, for all $i=1,2,3$. 
     Fix an arbitrary order $\preceq$ of the alphabet. Then  
  there must be two indices $i,j\in \{1,2,3\}$ such that $\zeta_i,\zeta_j$ are either both bigger or both smaller than 
 $\gamma$. By Theorem ~\ref{th:Wheeler},  we have   $\La\not \in \Wh(\preceq)$ and by the generality of the choice of $\preceq$ it follows $\La\notin \EW$.
\end{proof}

 \begin{lemma*}[\ref{lem:Sigma^*}] 
  $\La\in \DEF$   iff  $\La$ is finite,  or  $\La\in \UW$ and  $\pf\La=\Sigma^*$.   
\end{lemma*}
 
\begin{proof} 
If $\La\in \DEF$ then $L=F \cup \Sigma^*G$ for finite sets $F,G$.  Hence, $\pf\La=\Sigma^*$ if $\La$ is not finite. Moreover, $\La\in \UW$ by Lemma~\ref{lem:slt_are_uw} as definite languages are strictly locally testable.  

Conversely, suppose that $\La$ is finite or  $\La\in \UW$ and $\pf\La=\Sigma^*$. If $\La$ is finite then $\La\in \DEF$  and we are done. 
If $\La\in \UW$ and $\pf\La=\Sigma^*$, consider the minimum    automaton  $\mt D_\La=(Q,q_0, \dots)$ recognizing $\La$. Since $\pf\La=\Sigma^*$, the automaton $\mt D_\La$ is complete and   trimmed. We first prove that $\La\in \SLT$. Suppose, by way of a contradiction, that this is not the case. Then, by   Lemma~\ref{lem:caron},  we have that $\mt D_\La$ contains two states $p\neq q$ which are both the starting point of a cycle labeled by a string $\delta\in \Sigma^+$. Let $\alpha\neq  \beta$ be words reaching $p,q$ from the initial state, respectively. If $k$ is such that     $|\delta^k|>|\alpha|, |\beta|$, we are sure  that     no power of $\delta^k$ is  a suffix of    $\alpha$ or  $ \beta $.

Let $\gamma=\delta^k$ and consider the sequence $q_0,q_1,q_2, \dots, q_i, \dots$   that we obtain,  in $\mt D_\La$, reading the powers of $\gamma$ starting from the initial state
\[q_0 \rightarrow ^\gamma q_1 \rightarrow ^\gamma q_2  \rightarrow^\gamma \dots \rightarrow^\gamma q_i  \rightarrow^\gamma\dots \rightarrow^\gamma q_j  \rightarrow^\gamma\dots \]
(here we are using the hypothesis $\pf\La=\Sigma^*$: this path has to be in $\mt D_\La$).
 Then, there exist $i, j$ with $i<j$ such that $q_i=q_j$.  
 Notice that   $\rho=\gamma^{j-i}$ labels   cycles from the states $p,q$ and $r=q_i$;   if $\nu=\gamma^i$, by using an higher power of $\gamma$ if necessary, we may suppose that  $\rho$ is not a suffix of $\nu$. 
%Notice that $\nu\neq \alpha, \beta$ since, by construction, no power of $\gamma$ is  a suffix of    $\alpha$ or  $ \beta $.

 We now consider two cases:
 \begin{enumerate}
     \item $r\neq p, r\neq q$.   In this case  we have three       cycles labeled $\rho$ starting from      three different  $\mt D_\La$-states $p,q,r$.  By Lemma~\ref{lem:3cycles} we obtain     $\La\notin \EW$, contradicting $\La\in \UW$.   
     
 \item $r=p$ or $r=q$. Suppose w.l.o.g. that $r=p$.  In this case we have that $\nu=\gamma^i$ reaches $r$  and, since  $\rho$ is a bigger power of $\gamma $,    $\nu$ is always co-lexicographically smaller than $ \rho$ w.r.t. any order of the alphabet. Moreover, we know that $\rho$ is  not a suffix of $\beta$. This implies that we can always find an order $\preceq$ of the alphabet such that $\beta\prec \rho$: if $\beta$ is not  a suffix of $\rho$  (otherwise we already have $\beta \prec \rho$) then,  since $\rho$ is not a suffix of $\beta$ either,   we may consider       the first letters  $x,y$ from the right in $\beta, \rho$, respectively,   marking a difference between the two words;  we may then consider an order $\preceq$ having  $x\prec y$ in order to obtain $\beta\prec \rho$ and $\nu\prec \rho$ (the last one always holds).  
 
 Hence, we have two states $r,p$  which are both the starting point of a cycle labeled $\rho$ and are reached by strings smaller than $\rho$. This implies $\La\not \in \Wh(\preceq)$, contradicting $\La\in \UW$.
 \end{enumerate}

This proves that  $\La\in \SLT$, hence, there are  finite sets $H,K,W,F$ such that 
 \[\La=F\cup \left [(H\Sigma^* \cap \Sigma^* K) \setminus \Sigma^* W \Sigma^*\right ]. \]
 We now prove that, since $\pf\La=\Sigma^*$, we have $\La\in \DEF$.
 We have that  $\epsilon \not \in W$, otherwise $ \Sigma^* W \Sigma^*=\Sigma^*$, and  $L=F$ would be  a finite set, contradicting $\pf\La=\Sigma^*$.
 
To conclude the proof  we   check that  $\epsilon \in H$  (so that $H\Sigma^*=\Sigma^*$) and $W=\emptyset$. 
 Since    $\La\subseteq F \cup  H\Sigma^*$, all words in $\La$ which are longer than $\max\Set{|\alpha|}{\alpha \in F}$  must have have a prefix in the finite set $H$, contradicting $\pf\La=\Sigma^*$ unless $\epsilon \in H$.
Suppose, by way of a contradiction, that    $\xi \in W$, and     consider a word $\beta$, longer than $\max\Set{|\alpha|}{\alpha \in F}$, such that   $\beta$ contains $\xi$ as a prefix. Hence,  $\beta \in W\Sigma^*$ and, since $\pf\La=\Sigma^*$, there exists  $\delta\in \Sigma^*$ such that $ \beta\delta \in \La$. However, since $ \beta\delta \not \in F$ (it is too long), we should  have $ \beta \delta \not \in  \Sigma^* W \Sigma^*$,   contradicting  $\beta \in W\Sigma^*$.
Hence, $W=\emptyset$ and  $ \La=F\cup  \Sigma^* K\in  \DEF$.   
\end{proof}

\begin{lemma*}[\ref{lemma:esiste}]
Let  $|\Sigma|\geq 2$ and let  $\La\in \UW$. If $\pf\La\neq \Sigma^*$, the following are equivalent:
\begin{enumerate}
    \item there exists an order $\preceq$ such that $\overline \La\not \in \Wh(\preceq)$;
\item   $ \pf\La\cap \pf{\overline\La}$ is an infinite set. 
   \end{enumerate}
\end{lemma*}
\begin{proof}
    ($1\Rightarrow 2$)~If there exists an order $\preceq$ such that $\overline \La\not \in \Wh(\preceq)$ but  $\La\in \UW$, then, in particular, $\La\in \Wh(\preceq)$.
    Since $\overline \La\not \in \Wh(\preceq)$, then  $\overline \La\not \in \SLT$. By Lemma~\ref{lem:caron}, $\mt  D_{\overline \La}$   must contain two equally labeled cycles (say, with label    $\gamma$) starting from two different states $p,q$. 
    One of them, say $p$, must be different from  the absorbing state $\overline q$ for $\overline \La$ (see the discussion on minimum complete or trimmed DFA for a language and its complement). Since     all words arriving in   $p$ are in $\pf\La\cap \pf{\overline\La}$,  we have  
      $\alpha \gamma^* \in \pf\La\cap \pf{\overline\La}$ and     $ \pf\La\cap \pf{\overline\La}$ is infinite.
%\item 

($2\Rightarrow 1$)~Suppose  $ \pf\La\cap \pf{\overline\La}$ is infinite.   Since  $\pf\La \neq \Sigma^* $ implies  
$\pf{\overline\La} \setminus \pf\La \neq \emptyset $, in $\mt D_{\overline \La} $ 
there is an absorbing final state $\overline q$. Let  $\beta$ be a word arriving in $\overline q$ in $\mt D_{\overline \La} $.  From the hypothesis we get that 
$\beta \neq \epsilon$ (otherwise $\overline {\La}=\Sigma^*$ would  contradict $ \pf\La\cap \pf{\overline\La}\neq \emptyset$),  and, since $ \pf\La\cap \pf{\overline\La}$ is infinite,    in $\mt D_{\overline {\La}}$ there must be   a cycle not passing through $\overline q$.  Let $p\neq \overline q$ be on this cycle,    $\alpha$  be a word arriving in $p$,   and  $\gamma\neq \epsilon$ be a label of  a  cycle starting from $p$ and avoiding $\overline q$. 
We may suppose w.l.o.g.\ that $\gamma$ is not a suffix of $\alpha$ (if so, just consider a power of $\gamma$ instead of $\gamma$). Notice that there is a cycle labeled $\gamma$ both from $p$ and from $\overline q$ ($\overline q$ is an absorbing state, hence there are cycles starting from $\overline q$ labeled by any word).

We claim  that there exists an alphabet order $\preceq$ such that $p\intro \overline q$.
Consider two cases:
\begin{enumerate}
 \item $\alpha \dashv \gamma$. Let  $g$ be the last letter of $\gamma$ and let  $x\in \Sigma$ with  $x\neq g$. Fix an order $\preceq $ of the alphabet such that $x\prec g$.  Since   $\alpha$ arrives in $p$  and $ \beta x$ arrives in $\overline q$   we have either $\alpha \prec \beta x$ or $\alpha \succ \beta x$. In the first case,   we may consider the words
$\alpha \prec \beta x \prec  \alpha \gamma$, while if $\alpha \succ \beta x$  
we may consider the words
$\beta x \prec \alpha \prec \beta x \gamma$.
In both cases the words witness that $p\intro \overline q$.
\item $\alpha \not \dashv \gamma$. Then there is an order $\preceq$ of the alphabet such that $ \gamma\prec \alpha$ holds. If $\preceq$ is such an order, then, since we already supposed $\gamma \not \dashv \alpha$,  we have $\beta \gamma \prec \alpha \prec \beta \alpha$, witnessing $p\intro \overline q$.
 \end{enumerate}
 This proves the claim. Since a cycle labeled $\gamma$ starts from both $p$ and $\overline q$,   in $\mt D^2_{\overline \La} $ there is a cycle starting from $(p,\overline q)$, with $p\neq \overline q$ and $p\intro \overline q$, and Lemma~\ref{lem:wcycle} applied to $\overline \La$ yields 
 $\overline \La\not \in \Wh(\preceq)$.  
 %\end{enumerate} 
 \end{proof}

 \section{Examples}
 
\begin{example} \label{ap:ex} In Fig. \ref{fig:example WDFA}  we consider an example of the differences between automata $\mt D^c_{\La}$, $\mt D_{\La}$, 
$\mt D^c_{\overline{\La}}$, and $\mt D_{\overline {\La}}$.
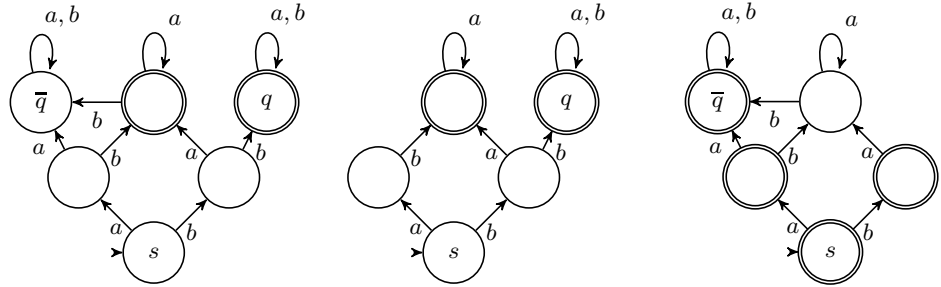
\begin{figure}[ht]%
\begin{center}
\begin{minipage}{0.26\textwidth}
\begin{tikzpicture}[->,>=stealth', semithick, initial text={}, auto, scale=.25]
 \node[state, label=above:{}, initial] (0) at (0,0) {$s$};
 \node[state, label=above:{}] (1) at (-4,4) { };
 \node[state, label=above:{}] (2) at (4,4) { };
 \node[state, label=above:{}, accepting] (3) at (0,8) {};
 \node[state, label=above:{}] (4) at (-6,8) {$\overline q$};
 \node[state, label=above:{}, accepting] (5) at (6,8) {$q$};
 
\draw (0) edge [below]  node  {$a$} (1);
 \draw (0) edge [below]   node   {$b$} (2);
 \draw (1) edge [below]  node   {$b$} (3);
 \draw (2) edge [below]  node    {$a$} (3);
 \draw (3) edge [loop above] node [above, xshift=8] {$a$} (3);
 \draw (1) edge   node [ xshift=-2, yshift=2]   {$a$} (4);
  \draw (4) edge [loop above] node [above, xshift=8] {$a,b$} (4);
   \draw (3) edge [below]  node     {$b$} (4);
    \draw (2) edge [below]  node [below, xshift=5, yshift=2] {$b$} (5);
  \draw (5) edge [loop above] node [above, xshift=8] {$a,b$} (5);

\end{tikzpicture}
\end{minipage}
\hspace{1.1cm}
\begin{minipage}{0.26\textwidth}
\begin{tikzpicture}[->,>=stealth', semithick, initial text={}, auto, scale=.25]
 \node[state, label=above:{}, initial] (0) at (0,0) {$s$};
 \node[state, label=above:{}] (1) at (-4,4) { };
 \node[state, label=above:{}] (2) at (4,4) { };
 \node[state, label=above:{}, accepting] (3) at (0,8) {};
 %\node[state, label=above:{}] (4) at (-6,8) {$\overline q$};
 \node[state, label=above:{}, accepting] (5) at (6,8) {$q$};
 
\draw (0) edge [below]  node  {$a$} (1);
 \draw (0) edge [below]   node   {$b$} (2);
 \draw (1) edge [below]  node   {$b$} (3);
 \draw (2) edge [below]  node    {$a$} (3);
 \draw (3) edge [loop above] node [above, xshift=8] {$a$} (3);

    \draw (2) edge [below]  node [below, xshift=5, yshift=2] {$b$} (5);
  \draw (5) edge [loop above] node [above, xshift=8] {$a,b$} (5);
 
\end{tikzpicture}
\end{minipage}
\hspace{1.1cm}
\begin{minipage}{0.26\textwidth}
\begin{tikzpicture}[->,>=stealth', semithick, initial text={}, auto, scale=.25]
 \node[state, label=above:{}, initial,accepting ] (0) at (0,0) {$s$};
 \node[state, label=above:{},accepting] (1) at (-4,4) { };
 \node[state, label=above:{},accepting] (2) at (4,4) { };
 \node[state, label=above:{} ] (3) at (0,8) {};
 \node[state, label=above:{}, accepting] (4) at (-6,8) {$\overline q$};

\draw (0) edge [below]  node  {$a$} (1);
 \draw (0) edge [below]   node   {$b$} (2);
 \draw (1) edge [below]  node   {$b$} (3);
 \draw (2) edge [below]  node    {$a$} (3);
 \draw (3) edge [loop above] node [above, xshift=8] {$a$} (3);
 \draw (1) edge   node [ xshift=-2, yshift=2]   {$a$} (4);
  \draw (4) edge [loop above] node [above, xshift=8] {$a,b$} (4);
   \draw (3) edge [below]  node     {$b$} (4);
 
\end{tikzpicture}
\end{minipage}
\end{center}
    \caption{\emph{Left:} The minimum complete DFA $\mt D^c_{\La}$ for a regular language $\La$, with both the absorbing state $\overline q$ for $\overline {\La}$  and the absorbing state $ q$ for $ {\La}$.   The minimum complete DFA $\mt D^c_{\overline{\La}}$ for $\overline{\La}$ is obtained by switching final and non final states.
    \emph{Center:}  The minimum trimmed DFA $\mt  D_{\La}$ for  $\La$, in which the absorbing state for $\overline{\La}$ disappears. 
    \emph{Right:}  The minimum trimmed DFA $\mt  D_{\overline{\La}}$ for  $\overline{\La}$, in which the absorbing state for ${\La}$ disappears. }
    \label{fig:example WDFA}
\end{figure}
\end{example}

\begin{example}\label{ex} We give an example of a language $\La$  such that $\pf{\La}=\Sigma^*$,    
 the  minimum trimmed automaton for $\La$  (on the left in the figure below)   satisfies the condition of  Theorem~\ref{th:Wheeler} (w.r.t. the usual order of the alphabet), proving that $\La\not \in \UW$,  while   the minimum trimmed automaton for $\overline \La$ (on the right)  does not contain two equally labeled cycles, proving that $\overline \La\in \UW$.   
 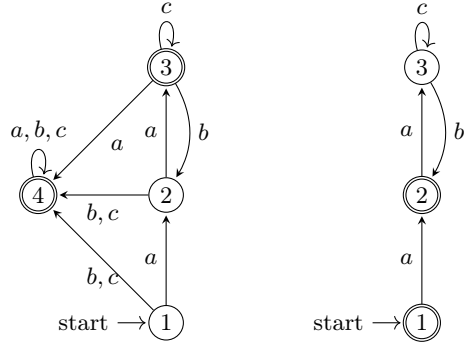
\begin{figure}[ht]
\centering
\begin{tikzpicture}[shorten >=1pt,node distance = 1.7cm, on grid, auto]
\tikzset{every state/.style={inner sep=2pt,minimum size=0pt}} 

\node (1) [state, initial] {1};
\node (2) [state,  above   = of 1] {2};
\node (3) [state, above  = of 2, accepting] {3};
\node (4) [state,    left = of 2, accepting] {4};
\node(5) [right =of 1] {};
  
\path [-stealth]
    (1) edge node {$a$} (2)
    (1) edge [below]  node{$b,c$} (4)
    (2) edge  node {$a$} (3)
     (2) edge   node {$b,c$} (4) 
     (3) edge [loop above ]  node {$c$} (3) 
     (4) edge [loop above ]  node {$a,b,c$} (4)
     (3) edge [bend left ]  node {$b$} (2)
     (3) edge [ ]  node {$a$} (4);
  
  \node (1') [state, initial, accepting,     right    = of 5 ] {1};
\node (2') [state, above    = of 1', accepting] {2};
\node (3') [state, above  = of 2'] {3};
  
\path [-stealth]
    (1') edge node {$a$} (2')
    (2') edge  node {$a$} (3')
     (3') edge [loop above ]  node {$c$} (3') 
     (3') edge [bend left ]  node {$b$} (2');  
\end{tikzpicture}
\caption{A language $\La$  (minimum trimmed DFA on the left)  with  $\pf\La=\Sigma^*$ which is not in $\UW$ while $\overline \La\in \UW$ (minimum trimmed DFA on the right)}\label{barLinUW}
\end{figure} 
\end{example}

 \section{Proofs from Section \ref{sec:decidingUW}}

\begin{lemma*} [\ref{lem:Pgeq2}]   
    Let $p, q \in Q$ with $p \neq q$. Then, $p \intr q$ if and only if $|P[p, q]| \geq 2$.
\end{lemma*}
\begin{proof}
    $(\Rightarrow)$ If $p \intr q$, then there exists an order $\preceq$ and strings $\alpha, \alpha' \in I_p$ and
    $\beta, \beta' \in I_q$ such that $\alpha \prec \beta$ and $\beta' \prec \alpha'$.
    If $\alpha$ is a suffix of $\beta$, then $(\dashv, \dashv) \in P[p, q]$. Otherwise, there exist 
    $a, b \in \Sigma$, with $a \prec b$, and a string $\gamma \in \Sigma^*$ such that $a\gamma \dashv \alpha$ and $b\gamma \dashv \beta$, and thus $(a, b) \in P[p,q]$. Similarly, for $\beta' \prec \alpha'$ we have either 
    $(\vdash, \vdash) \in P[p,q]$ or there exists  $a', b' \in \Sigma$, with $b' \prec a'$  and $(a', b')$ in $P[p, q]$. Furthermore, if such $a,b,a',b'$ exists, we must have $(a,b) \neq (a',b')$ because  
    $a \prec b$ and $b' \prec a'$. Hence, $|P[p,q]| \geq 2$.

$(\Leftarrow)$   Let $(a, b) \neq (a',b')$ be two elements in $P[p, q]$. 
Suppose  $\{(a, b), (a',b')\}\cap\{(\dashv, \dashv)\}=\emptyset$.  One of the pairs, say $(a,b)$,  must be different from $(\vdash, \vdash)$.   
 
  If $(a',b')=(\vdash, \vdash)$, then  any total order $\preceq$ on $\Sigma$ for which $a \prec b$ is such that 
 $p \intro q$. Otherwise, any total order $\preceq$ on $\Sigma$ for which $a \prec b$ and $b' \prec a'$ is such that 
 $p \intro q$. Similarly, if  $\{(a, b), (a',b')\}\cap\{(\vdash, \vdash)\}=\emptyset$ then there exists $\preceq$ such that         $p \intro q$.
Finally, if  $\{(a, b), (a',b')\}=\{(\dashv, \dashv), (\vdash, \vdash)\}$ then  $p\intro q$ holds for any order $\preceq$ of the alphabet $\Sigma$.
\end{proof}

\begin{lemma*}[\ref{lemma:P-table-geq2}]
    Let $p, q \in Q$, with $p \neq q$. Then $|P[p, q]| \geq 2$ iff $|P_{\lambda'}[p, q]| \geq 2$. 
\end{lemma*}

\begin{proof}
    $(\Rightarrow)$ We prove the statement by contraposition. Suppose $|P_{\lambda'}[p, q]| < 2$,  hence $(P_{\lambda'}[p, q] \setminus \{(\dashv, \dashv), (\vdash,\vdash)\})
    \subseteq \{(a, b)\}$, for some $a,b\in \Sigma$. Then, for every pair $p', q' \in Q$ such that there exists $\gamma \in \Sigma^*$
    for which $p = \delta(p', \gamma)$ and $q = \delta(q', \gamma)$ it holds $\lambda'(p') \subseteq \{a\}$ and
    $\lambda'(q') \subseteq \{b\}$. Now recall that $\lambda'(p')$ and $\lambda'(q')$ where chosen as a maximal subset of $\lambda(p')$ and $\lambda(q')$ with at most two elements and thus also   $\lambda(p') \subseteq \{a\}$ and
    $\lambda(q') \subseteq \{b\}$ must hold, from which it follows that $(P[p, q] \setminus  \{(\dashv, \dashv), (\vdash,\vdash)\})\subseteq \{(a, b)\}$. Then $|P[p, q]| < 2$ since 
    $(\dashv, \dashv)\in P[p, q] $ implies $(\dashv, \dashv)\in P_{\lambda'}[p, q] $ and the same holds for $(\vdash, \vdash)$. 

    $(\Leftarrow)$ Since $P_{\lambda'}[p, q] \subseteq P[p, q]$, the result is immediate. 
\end{proof}

\end{subappendices}

\paragraph{Funding.} Ruben Becker and Nicola Prezza are funded by the European Union (ERC, REGINDEX, 101039208). Views and opinions expressed are however those of the authors only and do not necessarily reflect those of the European Union or the European Research Council Executive Agency. Neither the European Union nor the granting authority can be held responsible for them. Giuseppa Castiglione is funded by the project “ACoMPA"  (CUP B73C24001050001),  NextGeneration EU programme PNRR ECS00000017 Tuscany Health Ecosystem (Spoke 6). Brian Riccardi received grants from the European Union's Horizon 2020 Research and Innovation Programme under the Marie Sk\l{}odowska-Curie grant agreement PANGAIA No. 872539, and from MUR 2022YRB97K, PINC, Pangenome INformatiCs: from Theory to Applications.
Alberto Policriti is founded by Project funded under the National Recovery and Resilience Plan (NRRP), Mission 4 
Component 2 Investment 1.4 – Call for tender No. 3138 of 16 December 2021, rectified by
Decree n.3175 of 18 December 2021 of Italian Ministry of University and Research funded by
The European Union – NextGenerationEU.
Award Number: Project code CN\_00000033, Concession Decree No. 1034 of 17 June 2022 
adopted by the Italian Ministry of University and Research, CUP G23C22001110007, 
Project title: “National Biodiversity Future Center – NBFC”.

\bibliographystyle{splncs04}
\bibliography{bib.bib}
\end{document}